\documentclass[ 10 pt, conference]{ieeeconf}  

\IEEEoverridecommandlockouts                              

\usepackage{amssymb}  
\usepackage{xcolor}
\usepackage{graphicx}
\usepackage{bbm}
 \usepackage{cite}

\usepackage{mathtools}

 \let\labelindent\relax
\usepackage{enumitem}

 \newtheorem{ass}{Assumption}[section]
 \newtheorem{prop}{Proposition}[section]
 \newtheorem{lem}{Lemma}[section]
 
 \newtheorem{coro}{Corollary}[section]
 
 \newtheorem{defn}{Definition}[section]
 \newtheorem{rem}{Remark}[section]
 \newtheorem{exam}{Example}[section]

\title{\LARGE \bf
A Behavioral Perspective on Models of Linear Dynamical Networks with Manifest Variables*
}

\author{Shengling Shi$^{1}$, Zhiyong Sun$^{2}$ and Bart {De Schutter}$^{1}$  {}
\thanks{*This research has received funding from the
European Research Council (ERC) under the European Union’s Horizon
2020 research and innovation programme (Grant agreement No.101018826 - CLariNet).  
}
\thanks{$^{1}$Shengling Shi and Bart De Schutter are with the Delft Center for Systems and Control, Delft University of Technology, the Netherlands
        {\tt\small \{s.shi-3, b.deschutter\}@tudelft.nl}}%
\thanks{$^{2}$Zhiyong Sun is with the Department of Electrical Engineering, Eindhoven University of Technology, the Netherlands
        {\tt\small z.sun@tue.nl}}%
 }

\begin{document}

\maketitle
\thispagestyle{empty}
\pagestyle{empty}

\begin{abstract}
Networks of dynamical systems play an important role in various domains and have motivated many studies on the control and analysis of linear dynamical networks. For linear network models considered in these studies, it is typically pre-determined what signal channels are inputs and what are outputs. These models do not capture the practical need to incorporate different experimental situations, where different selections of input and output channels are applied to the same network. Moreover, a unified view of different network models is lacking. This work makes an initial step towards addressing the above issues by taking a behavioral perspective, where input and output channels are not pre-determined. The focus of this work is on behavioral network models with only external variables. By exploiting the concept of hypergraphs, novel dual graphical representations, called system graphs and signal graphs, are introduced for behavioral networks. Moreover, connections between behavioral network models and structural vector autoregressive models are established. In addition to their connections in graphical representations, it is shown that the regularity of interconnections is an essential assumption when choosing a structural vector autoregressive model.
\end{abstract}

\section{INTRODUCTION}
Networks of dynamical systems are spatially distributed dynamical systems and consist of numerous individual subsystems that interact with each other to achieve sophisticated tasks. They play a pivotal role in various domains, including robotic swarms, power grids, transportation systems, and biomolecular networks \cite{mesbahi2010graph,ramos2022overview,majhi2022dynamics}.  

Motivated by the importance of these systems, considerable attention has been given to the development of analysis and control techniques for linear dynamical networks. These techniques have been developed in different modeling frameworks. Early works can be found in, e.g., the study of composite systems by Rosenbrock \cite{rosenbrock1974contributions} in the $70$s by using polynomial models. Without the intention to make a comprehensive review of the vast literature, we mention several other modeling frameworks: state-space models for linear dynamical networks \cite{fuhrmann2015mathematics,arcak2022compositional}, polynomial models \cite{fuhrmann2015mathematics}, the structural vector autoregressive (SVAR) model in time series analysis \cite{hyvarinen2010estimation}, and models consisting of transfer functions in system identification \cite{gonccalves2008necessary,shi2020generic}.

The above models of linear dynamical networks, if not autonomous, typically start with chosen input and output channels. However, from a practical point of view, it may be desired to conduct multiple experiments, with different choices of input and output channels, on the same dynamical network. The above freedom of choosing inputs and outputs of one network is not captured in the above modeling frameworks. Moreover, from a theoretical point of view, it is often not clear how these models are connected, and what underlying assumptions have been made when choosing a particular model. Overall, we argue that there is a lack of a unified theoretical framework for looking at these different linear network models.

The above two problems motivate us to adopt the perspective from the behavioral system theory, introduced by Willems \cite{willems1989models,willems1991paradigms}. In the behavioral theory, possible time trajectories in $(\mathbb{R}^q)^{\mathbb{Z}_+}$ allowed by a system, also called behaviors, take the central role. Different models that can generate the behaviors are regarded as different representations of the system. One of the most novel aspects in the behavioral theory is that inputs and outputs are not pre-distinguished but deduced from the modeling framework. This shows the potential of the behavioral theory to capture different experimental settings in modeling linear dynamical networks. Another important feature of the behavioral theory is that an axiomatic characterization of linear system representation is provided: A subset of $(\mathbb{R}^q)^{\mathbb{Z}_+}$ is closed, linear, and shift-invariant (time-invariant) iff it can be represented by a polynomial model \cite{willems1986time}. This brings polynomial models onto the central stage, and their connections with other representations have also been established \cite{willems1991paradigms}. This axiomatic characterization is particularly suitable for understanding the assumptions and connections of different network models.

Since the initial discussions on networks from the behavioral perspective \cite{willems2007behavioral}, several problems have been studied in the behavioral theory \cite{bisiacco2016consensus,shali2022composition,steentjes2022canonical} and in its applications to physical networks \cite{hughes2017controllability}. The results in \cite{bisiacco2016consensus,steentjes2022canonical} concern control problems of linear dynamical networks from the behavioral perspective. In \cite{shali2022composition}, the focus is on formalizing additional restrictions on the behavior, called contracts, for input-output models. The analysis of passive networks and electrical circuit networks is studied in \cite{hughes2017controllability}. Another interesting extension of the behavioral theory is the incorporation of category theory \cite{fong2016categorical}. However, we are not aware of recent works that connect different network models and study network representations in a behavioral framework.

In this work, following \cite{willems1991paradigms,willems2007behavioral},  we make an initial step towards addressing the two ultimate issues discussed previously, i.e., the incorporation of freedom for choosing input and output channels in a network modeling framework and building a unified view on different network models. For simplicity, we limit the scope to linear network models without latent variables, i.e., the variables that are internal and unmeasured. Compared to the previous studies on networks in the behavioral setting, our contributions are as follows:
\begin{itemize}
    \item Novel graphical representations for the network model in the behavioral setting are introduced;

    \item An explicit connection between the behavioral network model and the SVAR model is established.
\end{itemize}

More specifically, we start from the interconnection of linear dynamical systems in the behavioral framework \cite{willems1991paradigms}, where inputs and outputs are not pre-distinguished. Then, novel graphical representations of behavioral networks are introduced by exploiting the concept of \textit{hypergraphs} \cite{schrijver2003combinatorial}. This leads to novel dual concepts called \textit{system graphs} and \textit{signal graphs} as formal graphical representations of behavioral networks. The system graph models dynamical systems as vertices and formalizes the standard graphical visualization of networks in the behavioral theory. The signal graph models signals as vertices and systems as edges, which match more closely with the behavioral framework, where trajectories play the central role. Finally, we establish a connection between behavioral networks and SVAR models, due to the important role of SVAR models in various applications \cite{cologni2008oil}. We show that the regularity of interconnections, as introduced in \cite{willems1997interconnections}, is a fundamental assumption underlying SVAR models. More interestingly, we also show explicit connections between system graphs, signal graphs of behavioral networks and directed graphs of SVAR models. 

\section{Preliminaries}
\subsection{Notation}
The notation $\mathbb{Z}_+$ denotes the set of non-negative integers. The set of all functions from a set $\mathcal{X}$ to a set $\mathcal{Y}$ is denoted as $\mathcal{Y}^\mathcal{X}$. The notation $\mathrm{col}(F_1, F_2)$, where $F_1$ and $F_2$ are two matrices with the same number of columns, denotes $\begin{bmatrix}
    F_1\\F_2
\end{bmatrix}$. Given $w \in \mathbb{R}^q$, we say that $\mathrm{col}(u,y)$ is a component-wise partition of $w$ if a permutation matrix $\Pi$ exists such that $ w = \Pi \mathrm{col}(u,y)$. Given a positive integer $m$, the symbol $\mathbb{I}_m$ denotes the set $\{1,\dots,m\}$. The notation $w_{\mathbb{I}_m}$ denotes $\mathrm{col}(w_{1}, \dots, w_{m})$, and for any $\bar{\mathbb{I}} =\{i_1,\dots,i_r\} \subseteq  \mathbb{I}_m$, $\pi_{\bar{\mathbb{I}}}(w_{\mathbb{I}_m} )$ denotes a projection and equals $w_{\bar{\mathbb{I}}}$.

For any positive integer $q$, $\mathbb{R}^{\bullet \times q}[s]$ denotes the set of polynomial matrices with $q$ columns, an arbitrary number of rows, and an indeterminate $s$. A polynomial matrix $U \in \mathbb{R}^{q \times q}[s]$ is unimodular if its inverse is also a polynomial matrix. The notation $\mathrm{deg}(r)$ of a polynomial $r \in \mathbb{R}[s]$ denotes its degree. We define a function $\mathcal{M} : \mathbb{R}^{m \times q}[s] \to \{0,1\}^{m \times q}$ such that $[\mathcal{M} (R)]_{ij}=0$ iff $R_{ij} =0$, i.e., $\mathcal{M} (R)$ shows the sparsity pattern of $R$. Note that $R=0$ means $R$ is a zero function. Notation $[R]_{i\star}$ denotes the $i$-th row of $R$.
\subsection{Systems theory from a behavioral perspective}
Following \cite{willems1991paradigms}, we define a dynamical system as $\Sigma=(\mathbb{T}, \mathbb{W}, \mathcal{B})$ with a time axis $\mathbb{T}$, a signal space $\mathbb{W}$, and its behavior $\mathcal{B} \subseteq \mathbb{W}^\mathbb{T}$. In this work, we will focus on linear time-invariant systems in a discrete-time setting, and the main results also hold for the continuous-time setting with minor modifications as highlighted later in Remark~\ref{rem:CT}. 

Following the notation in \cite{willems1997interconnections}, any $R \in \mathbb{R}^{\bullet \times q}[s]$ defines a discrete-time dynamical system via the equation 
\begin{equation} \label{eq:Kernel}
   R(\sigma)w=0
\end{equation}
as $\Sigma igma  (R) \triangleq \big(\mathbb{Z}_+, \mathbb{R}^q, \mathrm{ker}(R)\big)$, where $\mathrm{ker}(R) = \{w \in (\mathbb{R}^q)^{ \mathbb{Z}_+} \mid w \text{ satisfies }\eqref{eq:Kernel} \}$,
and $\sigma$ is the shift operator, i.e., $\sigma w(t)=w(t+1)$. Notation $\mathcal{L}^q$ denotes the set of all such dynamical systems, i.e., $\mathcal{L}^q \triangleq \{ \Sigma igma(R) \mid R \in \mathbb{R}^{\bullet \times q}[s]  \}$. As shown in \cite{willems1991paradigms}, $\mathcal{L}^q$ actually contains the linear time-invariant finite-dimensional dynamical systems.

A polynomial matrix $R$ defines $\Sigma igma(R)$ uniquely; however, $\Sigma igma(R)$ does not uniquely define a polynomial matrix: It holds that $\Sigma igma(R) = \Sigma igma(U R)$, where $U$ is any unimodular polynomial matrix \cite{willems1986time}. We call $R \in \mathbb{R}^{\bullet \times q}[s]$ a \textit{kernel representation} of $\Sigma \in \mathcal{L}^q$ if $\Sigma = \Sigma igma(R)$. Moreover, any $\Sigma \in \mathcal{L}^q$ admits a kernel representation $R \in \mathbb{R}^{\bullet \times q}[s]$ that has full row rank (over the polynomial ring $\mathbb{R}[s]$). Such an $R$ is called a \textit{minimal kernel representation} of $\Sigma$, and its rank is uniquely determined by $\Sigma$. We define the function $p : \mathcal{L}^q \to \{1,\dots, q\}$ such that $p(\Sigma)$ is the rank of any minimal kernel representation of $\Sigma \in \mathcal{L}^q$. If $p(\Sigma) = q$, i.e., $\Sigma$ can be represented by $R(\sigma)w=0$ with $\det(R)\not=0$, $\Sigma$ is an autonomous system.

Given a minimal kernel representation $R \in \mathbb{R}^{p(\Sigma) \times q}[s]$ of $\Sigma \in \mathcal{L}^q$, the McMillan degree of $R$ is defined as the maximum degree of all the $p(\Sigma) \times p(\Sigma)$ minors of $R$ \cite{willems1997interconnections}, and it is uniquely determined by $\Sigma$. Therefore, we define the function $n : \mathcal{L}^q \to \mathbb{Z}_+$ such that $n(\Sigma)$ is the McMillan degree of $\Sigma$. As shown in \cite{willems1991paradigms}, $n(\Sigma)$ also equals the state dimension of a minimal state-space representation of $\Sigma$.

For any $\Sigma= \big(\mathbb{Z}_+, \mathbb{R}^q, \mathcal{B} \big) \in \mathcal{L}^q$, there exist $P \in \mathbb{R}^{p(\Sigma) \times p(\Sigma)}[s]$ with $\det(P) \not=0$ and $Q \in \mathbb{R}^{p(\Sigma) \times (q-p(\Sigma))}[s]$ such that 
\begin{equation} \label{eq:IOsystem}
\Sigma: \quad   P(\sigma)y = Q(\sigma) u
\end{equation}
describes the behavior of $\Sigma$, i.e., $\mathcal{B}= \{w \mid w = \Pi \mathrm{col}(u, y ) \in (\mathbb{R}^q)^{ \mathbb{Z}_+},  \eqref{eq:IOsystem} \} $ for some permutation matrix $\Pi$. Such a component-wise partition $\mathrm{col} (u, y )$ of $w$ is called an \textit{input-output partition} with input $u$ and output $y$. The input $u$ is \textit{free}, i.e., for any $u$, there exists a $y$ such that $ \Pi \mathrm{col}(u, y ) \in \mathcal{B}$. If $P^{-1}Q$ is proper, then $\mathrm{col} (u, y )$ is called a \textit{proper input-output partition} of $w$. A proper input-output partition always exists: Given a minimal kernel representation $R$ of $\Sigma \in \mathcal{L}^q$, choose $P$ to be a $p(\Sigma) \times p(\Sigma)$ submatrix of $R$ such that $\mathrm{deg}( \det(P))$ is the largest among all submatrices. However, note that this partition may not be unique, leading to the freedom in choosing different inputs and outputs.

\section{Manifest networks in a behavioral setting}
\subsection{Algebraic representation}
 In a network setting with only manifest variables, i.e. variables that are external or measured,  we are interested in the interconnection of several dynamical systems.  

\begin{defn}[\cite{willems1991paradigms}]
Consider $N$ dynamical systems $\Sigma_{i}=(\mathbb{T},\mathbb{W}, \mathcal{B}_i ) $ for $i \in  \mathbb{I}_N$, then their interconnected system is defined as $\Sigma = \land_{i=1}^N \Sigma_i \triangleq (\mathbb{T},\mathbb{W},\cap_{i=1}^N \mathcal{B}_i ) $, and $\Sigma_i$ is called a component of $\Sigma$.
\end{defn}

The definition of interconnected systems, also called dynamical networks, considers the total interconnection of dynamical systems, i.e., all $\Sigma_i$ have the same signal space. It covers the special cases where $\Sigma_i$ have different signal spaces, i.e., the so-called partial interconnection, as shown in the following example from \cite{willems1997interconnections}: The interconnection of $\Sigma_1 = (\mathbb{T}, \mathbb{W}_1 \times \mathbb{W}_2 , \mathcal{B}_1)$ and $\Sigma_2 = (\mathbb{T}, \mathbb{W}_2 \times \mathbb{W}_3 , \mathcal{B}_2)$ is equivalent to the interconnection of $\bar{\Sigma}_1 = (\mathbb{T}, \mathbb{W}_1 \times \mathbb{W}_2 \times \mathbb{W}_3, \mathcal{B}_1 \times (\mathbb{W}_3)^{\mathbb{T}})$ and $\bar{\Sigma}_2 = (\mathbb{T}, \mathbb{W}_1 \times \mathbb{W}_2 \times \mathbb{W}_3, (\mathbb{W}_1)^{\mathbb{T}} \times \mathcal{B}_2)$.
 
In this work, we will introduce a binary matrix and graphical representations to encode the above structure of a partial interconnection. To this end, we first define the following concept:
\begin{defn}
Given a dynamical system\footnote{We write the signal space as a product to decompose the signal vector $w$ into $L$ subvectors, i.e., $w = w_{ \mathbb{I}_L}= \mathrm{col}(w_1,\dots,w_L)$.} $\Sigma = (\mathbb{T},\prod_{ j \in \mathbb{I}_L}\mathbb{W}_j,  \mathcal{B} )$, the signal $w_j$ with  $j \in \mathbb{I}_L$ is said being \textit{unconstrained} in $\Sigma$ if $\mathcal{B} = \{w_{ \mathbb{I}_L}  \in (\prod_{ j \in \mathbb{I}_L}\mathbb{W}_j)^{\mathbb{T}}  \mid w_j \in (\mathbb{W}_j )^{\mathbb{T}}, w_{\mathbb{I}_L \setminus \{j\} } \in \pi_{\mathbb{I}_L \setminus \{j\} }(\mathcal{B})\} $.
\end{defn} 

Signal $w_j$ being unconstrained means that $\forall$ $w_j \in (\mathbb{W}_j )^{\mathbb{T}}$ and $\forall$ $ w_{\mathbb{I}_L \setminus \{j\} } \in \pi_{\mathbb{I}_L \setminus \{j\} }(\mathcal{B})$ lead to $w_{ \mathbb{I}_L} \in \mathcal{B}$. This is a stronger requirement than requiring $w_j$ to be \textit{free} \cite{willems1991paradigms}, i.e., $\forall$ $w_j \in (\mathbb{W}_j )^{\mathbb{T}}$, $\exists $  $w_{\mathbb{I}_L \setminus \{j\} } \in \pi_{\mathbb{I}_L \setminus \{j\} }(\mathcal{B})$ such that $w_{ \mathbb{I}_L} \in \mathcal{B}$. 

When the Euclidean space is considered as the signal space, we have the following algebraic characterization of the unconstrained signals:
\begin{lem} \label{lem:Uncons}
Given a dynamical system $\Sigma = \Big(\mathbb{Z}_+, \prod_{ j \in \mathbb{I}_L } \mathbb{R}^{q_j},  \mathcal{B} \Big) \in \mathcal{L}^q$ , then $w_j $ is unconstrained if and only if $R_j =0$ for any kernel representation $ [R_1\text{ }R_2 \text{ }\cdots R_L ] $ of $\Sigma$, where $R_j \in \mathbb{R}^{\bullet \times q_j}[s]$.
\end{lem}
\begin{proof}
The ``if" part is straightforward. We can prove the ``only if" part by contradiction as follows. Firstly, note that $0 \in  B= \mathrm{ker}\big([R_1\text{ }R_2 \text{ }\cdots R_L ]\big) $, and thus $w_{\mathbb{I}_L \setminus \{j\} } = 0 \in \pi_{\mathbb{I}_L \setminus \{j\} }(\mathcal{B})$. Assume that $R_j\not=0$, and consider $[R_1(\sigma)\text{ }\dots R_j(\sigma) \dots R_L(\sigma) ] w = 0$. There must exist $w_j \not=0$ such that $R_j(\sigma)w_j \not=0$. Then, it is clear $w = \mathrm{col}(0,\dots,0,w_j,0,\dots,0) \notin \mathcal{B}$, which contradicts $w_j$ being unconstrained.
\end{proof}

As shown in the above lemma, a signal vector being unconstrained simply means that the corresponding block matrix in the kernel representation is zero. If we consider a linear system $\Sigma = (\mathbb{Z}_+, \mathbb{R}^{q_1}\times \mathbb{R}^{q_2}\times \mathbb{R}^{q_3}, \mathcal{B}) \in \mathcal{L}^q$, then $w_2$ is unconstrained iff $\Sigma$ admits a kernel representation as 
$$
\begin{bmatrix}
 R_1(\sigma) & 0 & R_3(\sigma)  
\end{bmatrix} \begin{bmatrix}
    w_1 \\
    w_2\\
    w_3
\end{bmatrix} = 0,
$$
where $w_i \in (\mathbb{R}^{q_i})^{\mathbb{Z}_+}$. The above sparsity pattern of the kernel representation is invariant after an equivalent transformation, i.e., the pre-multiplication of a unimodular matrix.

With the above concept, we can characterize the sparsity pattern of interconnected systems by a binary matrix:
\begin{defn} \label{eq:Incidence}
Consider $N$ dynamical systems $\Sigma_i = (\mathbb{T},\prod_{ j \in \mathbb{I}_L}\mathbb{W}_j,  \mathcal{B}_i )$ for $i \in \mathbb{I}_N$, the 
 \textit{incidence matrix} $S \in \{0,1\}^{N \times L}$ of the interconnected system $\land_{i=1}^N \Sigma_i $ is defined as $S_{ij}=0$ iff $w_j$ is unconstrained in $\Sigma_i$. 
\end{defn}

An incidence matrix $S$ of an interconnected system contains important structural information: Its number of rows equals the number of components, and its number of columns equals the number of signal vectors. As an example, consider an interconnected system with $4$ components $\Sigma_i = ( \mathbb{Z}_+, \prod_{j=1}^4 \mathbb{R}^{q_j},  \mathcal{B}_i) \in \mathcal{L}^q $ for $i \in \mathbb{I}_4 $ and an incidence matrix 
\begin{equation} \label{eq:Sexam}  
S = \begin{bmatrix}
 1 & 1 & 1& 0 \\
 0 & 1 & 0& 1\\
 0& 0 & 1 & 1 \\
  0& 0 & 1 & 1
\end{bmatrix}.
\end{equation}
The above incidence matrix implies that $\land_{i=1}^4 \Sigma_i$ admits a kernel representation 
\begin{equation} \label{eq:exam1} 
\begin{bmatrix}
 R_{11}(\sigma) & R_{12}(\sigma) &  R_{13}(\sigma) & 0 \\
 0 & R_{22}(\sigma) &0 & R_{24}(\sigma)\\
 0 & 0 & R_{33}(\sigma) & R_{34}(\sigma) \\
  0 & 0 & R_{43}(\sigma) & R_{44}(\sigma) 
\end{bmatrix}\begin{bmatrix}
    w_1 \\
    w_2\\
    w_3 \\
    w_4
\end{bmatrix}=0,
\end{equation}
where the $i$-th block row is a kernel representation of $\Sigma_i$. In \eqref{eq:exam1}, $w_i$ can be scalar-valued or vector-valued, and each row can be a block row or one-dimensional. 
\begin{rem}
The incidence matrix of \eqref{eq:exam1} is invariant after the equivalent transformation of each component, i.e., pre-multiplication of \eqref{eq:exam1} by a block diagonal unimodular matrix. However, it may change after the equivalent transformation of the complete network, i.e., pre-multiplication of \eqref{eq:exam1} by a general unimodular matrix. 
\end{rem}

 We summarize the above point and the previous discussions into the following result:
\begin{coro}
Consider discrete-time systems $\Sigma_i  = ( \mathbb{Z}_+,\prod_{ j \in \mathbb{I}_L } \mathbb{R}^{q_j},  \mathcal{B}_i \big) \in \mathcal{L}^{q}  $ for $i\in \mathbb{I}_N$ and a matrix $S \in \{0,1\}^{N \times L}$, the following statements are equivalent:
\begin{enumerate}[label=(\alph*)]
    \item The matrix $S$ is an incidence matrix of $\land_{i=1}^N \Sigma_i$;

    \item It holds that $S_{ij} = 0$ iff $R_{ij}=0$ for any kernel representation $ \begin{bmatrix}
        R_{i1} & R_{i2} &\dots & R_{iL}
    \end{bmatrix} $ of $\Sigma_i$, where $R_{ij} \in \mathbb{R}^{\bullet \times q_j}[s]$.
\end{enumerate}
\end{coro}
 
\begin{rem}
 Given a system $\Sigma \in \mathcal{L}^q$, it can be decomposed as an interconnection of components, which is analogous to the step of zooming in modeling \cite{willems1991paradigms}. However, this decomposition is typically non-unique.  
\end{rem}
 \subsection{Graphical representation}
Models of networks are typically associated with a graphical representation, e.g., directed or undirected graphs. While the graphical visualization of interconnected systems has been used extensively in the behavioral theory \cite{ willems2007behavioral, shali2022composition, steentjes2022canonical}, it is typically not formalized in terms of graph theory. Graphical representations are considered in \cite{willems1991paradigms, bisiacco2016consensus} but with pre-determined inputs and outputs. Other graphical representations are mainly motivated by the modeling of electrical circuits \cite{willems2010terminals,hughes2017controllability} and are hard to connect to graphs of other linear network models.
 
 Since an incidence matrix $S$ has a one-to-one correspondence with the so-called hypergraph \cite{schrijver2003combinatorial}, we formalize the graphical representation of an interconnected system in this subsection. 

\begin{defn}[\cite{schrijver2003combinatorial}]
A \textit{hypergraph} is a pair $(\mathcal{V}, \mathcal{E})$ with a finite set $\mathcal{V}$ and a multiset \footnote{A multiset is a collection of elements with possible repeated elements.} $\mathcal{E} $ that contains subsets of $\mathcal{V}$.
\end{defn}

Each element in $\mathcal{V}$ is called a \textit{vertex} and each element in $\mathcal{E}$ is called a \textit{net} or a \textit{edge}. The nets are unordered sets and thus undirected. One net can contain multiple vertices or a single vertex, in contrast to standard graphs where an edge always connects two vertices\footnote{The edge set of a standard graph is typically a subset of $\mathcal{V} \times \mathcal{V}$ \cite{schrijver2003combinatorial}.}. This modeling feature is useful when cluster-wise interactions are of interest, e.g., in modeling consensus phenomena or contagion processes \cite{majhi2022dynamics}.

Any $S \in \{0,1\}^{N \times L} $ uniquely specifies a hypergraph via a function $\mathcal{G}(S) = ( \mathcal{V} , \mathcal{E})$, where
\begin{equation} \label{eq:InciGraph}
\mathcal{V}= \mathbb{I}_L, \text{ and } \mathcal{E}=\{E(i) \mid i \in \mathbb{I}_N \},
\end{equation}
and
$ E(i) \triangleq \{k \in \mathbb{I}_L \mid S_{ik} = 1 \}.$
The multiset $\mathcal{E}$ in \eqref{eq:InciGraph} contains repeated elements if $S$ contains identical rows. The \textit{dual hypergraph} of $\mathcal{G}(S)$ is defined as $\mathcal{G}(S^\top)$, where edges in $\mathcal{G}(S)$ become vertices in $\mathcal{G}(S^\top)$.

Given an incidence matrix $S$ of an interconnected system and the hypergraph $\mathcal{G}(S)$, each vertex in $\mathcal{G}(S)$ denotes a signal vector and each edge denotes a component of the interconnected system. This is in contrast to the standard graphical visualization in the behavioral theory \cite{willems1997interconnections, shali2022composition, steentjes2022canonical}, where components are drawn as vertices and signals are drawn as edges. This type of graph can actually be obtained as the dual hypergraph $\mathcal{G}(S^\top)$. We distinguish the above two different graphical representations in the following definition:
\begin{defn} \label{def:SysSigGraph}
Given an incidence matrix $S \in \{0,1\}^{N \times L}$ of an interconnected system $\land_{i=1}^N \Sigma_i$, where $\Sigma_i = (\mathbb{T},\prod_{ j \in \mathbb{I}_L}\mathbb{W}_j,  \mathcal{B}_i )$, the hypergraph $\mathcal{G}(S)$ is called the \textit{signal graph} of $\land_{i=1}^N \Sigma_i$, and the dual hypergraph $\mathcal{G}(S^\top)$ is called the \textit{system graph} of $\land_{i=1}^N \Sigma_i$.
\end{defn}

\begin{exam}
Consider the interconnected system \eqref{eq:exam1} of $4$ components and its incidence matrix \eqref{eq:Sexam}. The incidence matrix leads to a signal graph $(\mathcal{V},\mathcal{E})$, with $\mathcal{V} = \mathbb{I}_4$ and the multiset $\mathcal{E} = \{ \mathbb{I}_3, \{2,4\}, \{3,4\},\{3,4\} \}.$ The signal graph $\mathcal{G}(S)$ is shown in Fig.~\ref{fig:exam0}(a), where we label vertex $i$ as $w_i$ equivalently, and each net represents one row of $S$ and thus a component $\Sigma_i$. A net is visualized as a cluster in the signal graph, while a two-vertex net can also be visualized as an undirected edge. Note that the net $\{3,4\}$ appears twice in $\mathcal{E}$ due to the last two identical rows in \eqref{eq:Sexam}. The system graph $\mathcal{G}(S^\top)$ has $\mathcal{V} = \mathbb{I}_4$ and $\mathcal{E} = \{ \{1\}, \{1,2 \}, \{1,3,4\},\{2,3,4\}\}$. It is visualized in Fig.~\ref{fig:exam0}(b), where we have labelled the vertex $i \in \mathcal{V}$ as $\Sigma_i$. The nets represent the rows of $S^\top $ and thus the signals. They are drawn as edges, following the conventional visualization in the behavioral theory \cite{willems1997interconnections, shali2022composition, steentjes2022canonical}. For example, the net $\{2,3,4\}$ represents the signal $w_4$ and is drawn as an edge connecting the components $\Sigma_2$, $\Sigma_3$, and $\Sigma_4$.
\begin{figure}[h]
\begin{minipage}{0.23\textwidth}
\centering
\includegraphics[scale=0.4]{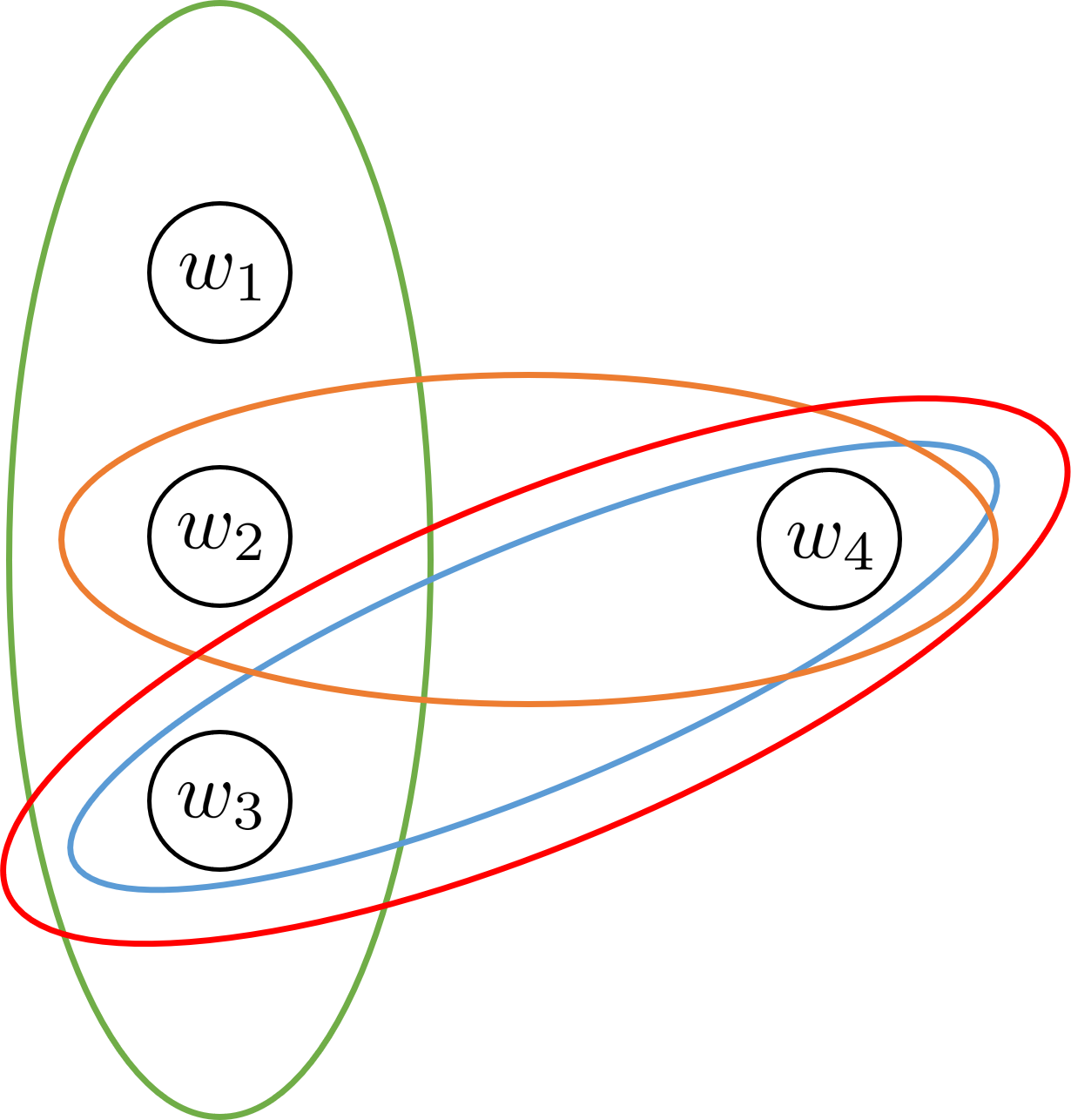}
\\(a)
\end{minipage}
\begin{minipage}{0.23\textwidth}
\vspace{0.5cm}
\centering
\includegraphics[scale=0.4]{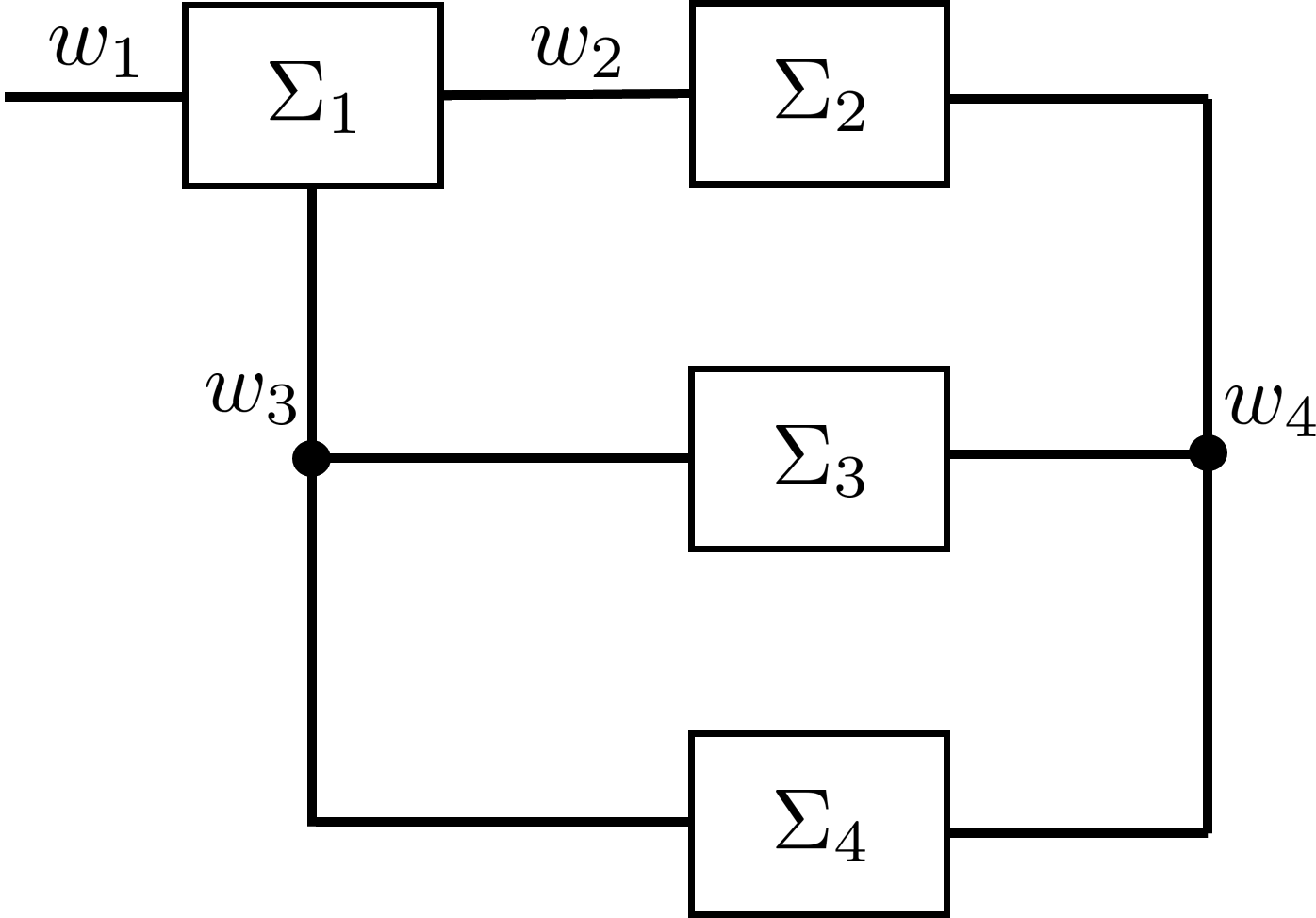}
\vspace{0.1cm}
\\(b)
\end{minipage}
\caption{Two graphical representations of \eqref{eq:exam1} are shown. The signal graph in (a) represents the $4$ signal vectors as $4$ vertices and the $4$ components as $4$ nets/edges. The dual graph in (b) represents the $4$ components as vertices and the $4$ signals as nets, visualized as edges by convention. }
\label{fig:exam0}
\end{figure}
\end{exam}

While block diagrams like system graphs and Fig.~\ref{fig:exam0}(b) are commonly used in systems and control theory, signal-based graphs like Fig.~\ref{fig:exam0}(a) are typically used in system identification \cite{shi2020generic} and machine learning \cite{hyvarinen2010estimation}. The duality in Definition~\ref{def:SysSigGraph} will be essential to connect the above different network models from various domains. A preliminary study of this direction will be given later in Section~\ref{sec:VARX}.
\begin{rem}
In \cite{chetty2015network}, relevant concepts called ``subsystem structure" and ``signal structure" are discussed, where the first one is associated with the state-space representation and the second one is associated with models containing transfer functions. By contrast, the dual graphs in this work are representations of kernel representations. 
\end{rem}

\subsection{Regularity of interconnection}
For the interconnection of dynamical systems, the concept of regularity has been studied in the behavioral theory \cite{willems1991paradigms,willems1997interconnections}. The regularity is typically studied for the interconnection of two dynamical systems in the context of control, i.e., the interconnection of a plant and a controller \cite{willems1997interconnections,julius2005canonical,trentelman2011optimal} or the interconnection of two local controllers in distributed control \cite{steentjes2022canonical}. In this work, we focus on the regularity of the interconnection of an arbitrary finite number of dynamical systems, as shown in the following definition.
\begin{defn}
Given $\Sigma_i \in \mathcal{L}^{q}$ for $i \in \mathbb{I}_N$, the interconnection $\Sigma= \land_{i=1}^N \Sigma_i$ is a \textit{regular interconnection} if 
$
p(\Sigma)= \sum_{i=1}^N p(\Sigma_i).
$
It is called a \textit{regular feedback interconnection} if it is a regular interconnection and $n(\Sigma) = \sum_{i=1}^N n(\Sigma_i).$
\end{defn}

A regular interconnection requires each component to have distinct and independent outputs. It also means that each component has independent physical constraints/laws to characterize its behavior. This interpretation is from the fact that, an interconnection $ \land_{i=1}^N \Sigma_i $ is regular iff $\mathrm{col}(R_1, \dots, R_N)$ has full row rank, where $R_i$ is a minimal kernel representation of $\Sigma_i$.

The regular interconnection of a plant and a controller is shown to be equivalent to a feedback connection, where the controller is a possibly non-proper input-output system \cite{willems1997interconnections}. To ensure that the controller is proper, a regular feedback interconnection is needed \cite{willems1997interconnections}. The regular feedback interconnection is defined by the condition of regular interconnection and the additional condition that the McMillan degrees of components respect the additive rule. The latter condition requires the components to have independent state variables. This is commonly satisfied by networks of physically decoupled systems, e.g., multiple robots, but may be violated by physical coupling, where states of different components are coupled via algebraic constraints, as shown later in Example~\ref{exm:circuit}. 

With the above concepts, the following two results are straightforward extensions of \cite[Th. 8 and 12]{willems1997interconnections} from two components to a finite number of components.
\begin{coro}
     Given an interconnected system $\Sigma = \land_{i=1}^N \Sigma_i$ with $\Sigma_i \in \mathcal{L}^{q}$ for all $i \in \mathbb{I}_N$, then its interconnection is a regular interconnection if and only if \footnote{The ``if" part is not stated in \cite{willems1997interconnections} but holds trivially. This also holds similarly for Corollary~\ref{lem:RegularNetwork}. } the signal vector $w$ of $\Sigma$ admits a component-wise partition $\mathrm{col}(y_1,\dots,y_N,u)$, where $y_i \in (\mathbb{R}^{p(\Sigma_i)})^{\mathbb{T}}$, such that
 \begin{itemize}
     \item in $\Sigma$, $\mathrm{col}(u,\mathrm{col}(y_1,\dots,y_N  ))$ is a proper input-output partition;

     \item in $\Sigma_i$, $\mathrm{col}\big(\mathrm{col}(u, y_1,\dots, y_{i-1}, y_{i+1},\dots y_N  ),y_i\big)$ is an input-output partition, for all $i \in \mathbb{I}_N$.
 \end{itemize}
\end{coro}

The above result states that if the interconnection is regular, the interconnected system can be regarded as an interconnection of (possibly non-proper) input-output systems with output-to-input interconnections: Each component receives the outputs of other components as its inputs. This is a fundamental property of many input-output network models.

To further ensure that all the input-output components are proper, a regular feedback interconnection should be considered:
\begin{coro}\label{lem:RegularNetwork}
 Given $\Sigma = \land_{i=1}^N \Sigma_i$ with $\Sigma_i \in \mathcal{L}^{q}$ for all $i \in \mathbb{I}_N$, then its interconnection is a regular feedback interconnection if and only if the signal vector $w$ of $\Sigma$ admits a component-wise partition $\mathrm{col}(y_1,\dots,y_N,u)$, where $y_i \in (\mathbb{R}^{p(\Sigma_i)})^{\mathbb{T}}$, such that
 \begin{itemize}
     \item in $\Sigma$, $\mathrm{col}(u, \mathrm{col}(y_1,\dots,y_N  ))$ is a proper input-output partition;

     \item in $\Sigma_i$, $\mathrm{col}\big(  \mathrm{col}(u, y_1,\dots, y_{i-1}, y_{i+1},\dots y_N  ),y_i\big)$ is a proper input-output partition, for all $i \in \mathbb{I}_N$.
 \end{itemize}
\end{coro}

Regularity of interconnections is a fundamental property in network models with output-to-input interconnections. It is typically assumed implicitly when a model is chosen. We will discuss this point for SVAR models in the next section.
\subsection{Lack of regularity}
Despite the existing discussions and applications of regular interconnections \cite{willems1997interconnections,bisiacco2016consensus,steentjes2022canonical}, there seems a lack of discussions on the situations where interconnections are not regular. It is easy to find real-world networks with non-regular (feedback) interconnections. Thus, it is of interest to investigate how to handle non-regular interconnections and to ask whether network models with output-to-input interconnections can still be applied to these situations. We motivate these discussions in the following example:
\begin{exam} \label{exm:circuit}
For the simple circuit in Fig.~\ref{fig:circuit}, we choose $V$, $V_C$, $I_1$, and $I_2$ to be the manifest variables, which denote the voltage across the circuit, the voltage across the capacitors, the current through the capacitor $C_1$, and the current through $C_2$, respectively. Let $d/dt$ denote the differentiation operator\footnote{The continuous-time setting does not affect the discussion here.}, and the models are $\Sigma_1: V - V_C = (d/dt) 	L (I_1+I_2) $, $\Sigma_2: I_1 =  (d/dt) C_1 V_C $, and $\Sigma_3: I_2 =(d/dt) C_2  V_C $.
\begin{figure}
    \centering
    \includegraphics[width=0.25\textwidth]{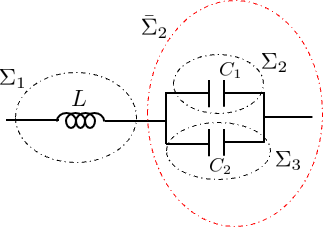}
    \caption{The interconnection $\land_{i=1}^3 \Sigma_i$ is not a regular feedback interconnection; however, if merging $\bar{\Sigma}_2 =\Sigma_2\land \Sigma_3$ into a single component, $\Sigma_1 \land \bar{\Sigma}_2$ becomes a regular feedback interconnection of two components.}
    \label{fig:circuit}
\end{figure}
We have $p(\Sigma_i) =1 $ for all $i \in \mathbb{I}_3$, and it is easy to verify $p( \land_{i=1}^3 \Sigma_i   ) =3$. This shows $p( \land_{i=1}^3 \Sigma_i  ) = \sum_{i=1}^ 3p(\Sigma_i)=3$, and thus, $\land_{i=1}^3 \Sigma_i$ is a regular interconnection. However, $n(\Sigma_i) = 1 $ for $i \in \mathbb{I}_3$, and it is easy to show $\sum_{i=1}^3 n(\Sigma_i) = 3 > n( \land_{i=1}^3 \Sigma_i)=2$. Therefore, $\land_{i=1}^3 \Sigma_i$ is not a regular feedback interconnection. This is because $\Sigma_2$ and $\Sigma_3$ share the same state $V_C$.

However, if we regard $\bar{\Sigma}_2 = \Sigma_2 \land \Sigma_3$ as a single component, then $p(\Sigma_1 \land \bar{\Sigma}_2 ) = p(\Sigma_1)+p(\bar{\Sigma}_2) =3 $ and $n(\bar{\Sigma}_2) + n(\Sigma_1) = n(\Sigma_1 \land \bar{\Sigma}_2 )=2$. This shows that $\Sigma_1 \land \bar{\Sigma}_2 $ is a regular feedback interconnection of two components, where $V_C$ and $I_2$ can be chosen as the outputs of $\bar{\Sigma}_2$ and are also inputs to $\Sigma_1$. Here, we obtain an output-to-input interconnection between $\bar{\Sigma}_2$ and $\Sigma_1$. \hfill $\blacksquare$
\end{exam} 

As shown in the above example, the lack of regular feedback interconnection is typically caused by static mappings between state variables of different components. One way to resolve the lack of regularity is by grouping some components into a single component, thus leading to the loss of structural information of the network. We can formalize this in the following result:
\begin{prop}
Given $\Sigma = \land_{i=1}^N \Sigma_i$ with $\Sigma_i \in \mathcal{L}^{q}$ for $i \in \mathbb{I}_N$, there exists a partition $\{ \bar{I}_1,\dots, \bar{I}_k\}$ of $\mathbb{I}_N$ such that $\land_{i=1}^k \bar{\Sigma}_i$ is a regular (feedback) interconnection of $k$ components, where $\bar{\Sigma}_i = \land_{j \in \bar{I}_i } \Sigma_j$.
\end{prop}
 \begin{proof}
     The existence is proved by letting $\bar{I}_1 = \mathbb{I}_N$.
 \end{proof}

In a kernel representation of a network, the partition of $\mathbb{I}_N$ corresponds to the merging of some block rows into a new block row. Non-trivial partitions are preferred in practice. Particularly, to maintain the structural information of the network, it may be desired to make $k$ as close to $N$ as possible. It is attractive to develop algorithms to obtain such a ``maximally merged" network, and this will be pursued in the future work.

\begin{rem}  \label{rem:CT}
 All the previous results for discrete-time systems also hold for continuous-time systems, if we replace the time set $\mathbb{Z}_+$ by $\mathbb{R}$, the shift operator $\sigma$ by the differentiation operator $d/dt$, and the definition of $\mathrm{ker}(R)$ by the set of weak solutions $w$ of $R(d/dt)w=0$ \cite{willems1997interconnections}.
\end{rem}

\section{ Structural VAR model} \label{sec:VARX}
\subsection{From SVAR to behavioral networks}
The SVAR model has been applied to several applications \cite{cologni2008oil,hyvarinen2010estimation}. This model can be justified from a statistical perspective \cite{hamilton1994time}. Having established the model of interconnected systems in the behavioral setting, we will investigate its connection with the SVAR model from a behavioral perspective. Another goal is to understand the underlying assumptions when choosing an SVAR model. In this work, we consider the SVAR model in a deterministic setting.

We consider the SVAR model \cite{hyvarinen2010estimation} with exogenous inputs\footnote{A typical example of the $i$-th row in a VAR model is $y_i(t) = ay(t-1) + by(t-2)$, with $a$ and $b$ being row vectors. We have shifted it in time to avoid the negative powers of the shift operator.}: For $i \in \mathbb{I}_N$,
\begin{align*}
y_i(t+l_i) =& X_{i,0} y(t+l_i) + X_{i,1} y(t+l_i-1)+ \dots + \\
& X_{i,l_i} y(t) +  Q_{i,0} u(t+l_i) + \dots Q_{i,l_i} u(t), 
\end{align*}
where $l_i \in \mathbb{Z}_+$, $y \in (\mathbb{R}^{N})^{\mathbb{Z}_+}$, $ u \in (\mathbb{R}^{m})^{\mathbb{Z}_+}$,  $X_{i,k} \in \mathbb{R}^{1 \times N}$ and the $i$-th element of $X_{i,0}$ is zero, and $Q_{i,k} \in \mathbb{R}^{1 \times m}$.  The model can be written more compactly as 
\begin{equation} \label{eq:VAR}
 X(\sigma)  y = Q(\sigma) u,
\end{equation}
which satisfies the following assumption on the instantaneous effects $ X_{i,0}$:
\begin{ass} \label{ass:SVAR}
The leading row coefficient matrix of $X$, i.e., $\mathrm{col}(\mathrm{e}_1- X_{1,0},\dots,\mathrm{e}_N- X_{N,0})$, has ones on its main diagonal and has full rank\footnote{Vector $\mathrm{e}_i$ is a standard basis row vector with the $i$-th element being one.}. 
\end{ass}

We note that Assumption~\ref{ass:SVAR} is less restrictive than the standard assumption in \cite{hyvarinen2010estimation}. In addition, $X$ and $Q$ may contain zeros to encode the interconnections among the scalar-valued signals.

There are two ways to interpret \eqref{eq:VAR}, either as a single dynamical system or as an interconnection of single-output components. The first view has been discussed in \cite{willems1991paradigms}, and we take a network perspective here.

Equation \eqref{eq:VAR} defines an interconnection of dynamical systems $\land_{i=1}^N \Sigma_i$, where the component $\Sigma_i$ is defined by the $i$-th row of \eqref{eq:VAR} with output $y_i$. The inputs of $\Sigma_i$ consist of the external input $u$ and the outputs of the other components. Moreover, the sparsity pattern of the matrices in \eqref{eq:VAR} specifies the incidence matrix of the interconnected system. The above discussions are summarized in the following result:
\begin{lem} \label{lem:VAR}
 Given \eqref{eq:VAR} that satisfies Assumption~\ref{ass:SVAR}, denote the $i$-th row of the model by $R_i=[X  \text{ } -Q    ]_{i \star}$, and consider dynamical systems $\Sigma_i = \Sigma igma( R_i   )$ for $i \in \mathbb{I}_N$. Then $\land_{i=1}^N \Sigma_i$ has an incidence matrix $\mathcal{M} ( [X   \text{ } -Q      ]   )$ and is a regular feedback interconnection. 
\end{lem}
\begin{proof}
The incidence matrix follows from the definition. From the model, it is clear that $p(\Sigma_i)=1$ and $n(\Sigma_i) =l_i$ for all $i \in \mathbb{I}_N$, leading to $\sum_{i=1}^N p(\Sigma_i) = N$ and $\sum_{i=1}^N n(\Sigma_i) = \sum_{i=1}^N l_i$. Since the leading row coefficient matrix of $X$ has full row rank, we have (i) $\det(X) \not=0$ and thus $p(\land_{i=1}^N \Sigma_i  ) = N = \sum_{i=1}^N p(\Sigma_i)$; (ii) $\mathrm{deg}(\det(X ) ) = \sum_{i=1}^N l_i$ and thus $n(  \land_{i=1}^N \Sigma_i   )=  \sum_{i=1}^N n(\Sigma_i) =\sum_{i=1}^N l_i$.
\end{proof} 

The above result shows that the regularity of interconnections is an essential property of the SVAR model.

From a graphical point of view, the incidence matrix $\mathcal{M} ( [X  \text{ } -Q  ] )$ induces a square \textit{adjacency matrix}:
\begin{equation} \label{eq:AdjVAR} 
A_{\mathrm{svar}} = \mathcal{M} \Big( \begin{bmatrix}
 X   & -Q  \\
 0_{m\times N} &  0_{m\times m}  
\end{bmatrix} \Big). 
\end{equation}
Matrix $A_{\mathrm{svar}}$ of \eqref{eq:VAR} is typically used to obtain a directed graph of the SVAR model \cite{eichler2013causal} via a function $\mathcal{G}_{\mathrm{d}}(A_{\mathrm{svar}})= (\mathcal{V}_{\mathrm{d}},\mathcal{E}_{\mathrm{d}})$, where the vertex set $\mathcal{V}_{\mathrm{d}} = \mathbb{I}_{N+m}$ denotes all the scalar-valued signals, and the edge set
$
\mathcal{E}_{\mathrm{d}} = \{(i,j) \in  \mathbb{I}_{N+m} \times \mathbb{I}_{N+m}    \mid  [A_{\mathrm{svar}} ]_{ji} \not=0 , \text{ } i \not= j \}.
$
Element $(i,j) \in \mathcal{E}_{\mathrm{d}} $ denotes a directed edge from vertex $i$ to $j$. 

The relation between the incidence matrix $S = \mathcal{M} ( [ X   \text{ } -Q ] )$ and the adjacency matrix $A_{\mathrm{svar}}$ leads to an immediate relation between the hypergraph $\mathcal{G}(S)$ and the directed graph $\mathcal{G}_{\mathrm{d}}( A_{\mathrm{svar}})$ of \eqref{eq:VAR}, as discussed in the following example:
\begin{exam}
Consider a SVAR model
$$
\begin{bmatrix}
   \sigma^l y_1\\
  \sigma^l  y_2\\
 \sigma^l   y_3
\end{bmatrix}= \begin{bmatrix}
0 & X_{12}(\sigma) & 0\\
X_{21}(\sigma) & 0 & 0\\
0 & 0 & 0
\end{bmatrix}\begin{bmatrix}
   y_1\\
   y_2\\
   y_3
\end{bmatrix} + \begin{bmatrix}
    Q_1(\sigma) \\
    0 \\
    Q_3(\sigma)
\end{bmatrix}u,
$$
which can be represented by a directed graph as in Fig.~\ref{fig:Var}(a), i.e., signals are vertices and $y_i$ has a directed edge to $y_j$ iff $X_{ji}\not=0$. Edges from $u$ to $y_j$ are defined similarly. 

\begin{figure}[h]
    \centering
    \includegraphics[scale=0.42]{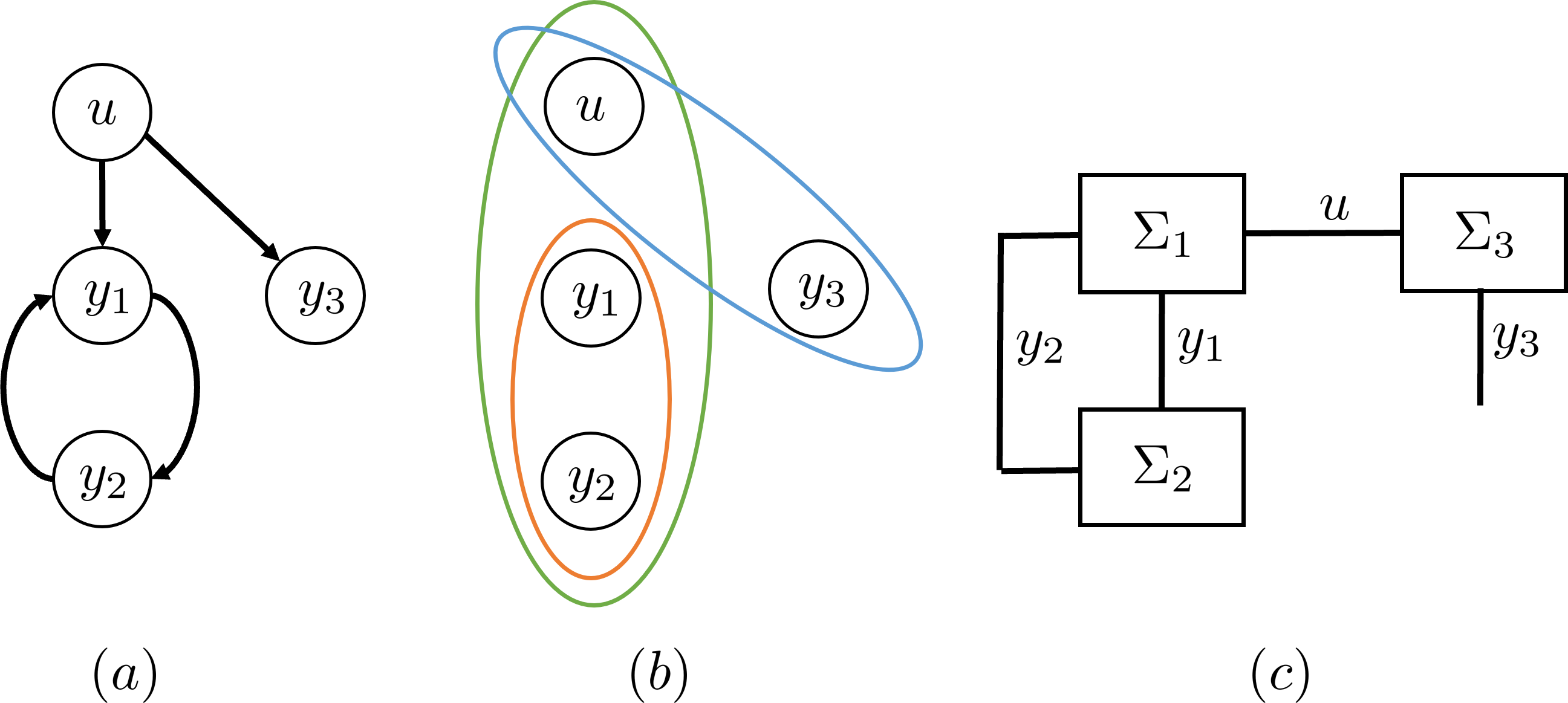}
    \caption{A directed graph of an SVAR model is shown in (a). The model also defines an interconnected system in the behavioral setting with a signal graph (b) and a system graph (c).}
    \label{fig:Var}
\end{figure}

This model also defines an interconnected system $\land_{i=1}^3 \Sigma_i$ of $3$ components with an incidence matrix: 
$$
S= \begin{bmatrix}
 1 & 1 & 0 &  1 \\
1 & 1 & 0 &   0   \\
  0 & 0& 1  &  1\\
\end{bmatrix}.
$$
The signal graph $\mathcal{G}(S)$ and the system graph $\mathcal{G}(S^\top)$ of $\land_{i=1}^3 \Sigma_i$ are shown in Fig.~\ref{fig:Var}(b) and (c), respectively. For example, $\Sigma_1$ in Fig.~\ref{fig:Var}(c) is defined by the first row of the SVAR model, and its behavior involves $y_1$, $y_2$, and $u$.

\end{exam}
\subsection{From behavioral networks to SVAR}
Conversely, an interconnection of dynamical systems can also be represented by \eqref{eq:VAR} if it is a regular feedback interconnection, as shown in the following result:
\begin{prop} \label{prop:VAR} 
 Consider $\Sigma_i = (\mathbb{Z}_+, \mathbb{R}^q, \mathcal{B}_i) \in \mathcal{L}^q$ with $p(\Sigma_i)=1$ for $i \in \mathbb{I}_N$. If $\land_{i=1}^N \Sigma_i$ is a regular feedback interconnection, there exist a proper input-output partition $\mathrm{col}(u,y)$ of $w$, $X \in \mathbb{R}^{N \times N}[s] $, and $Q \in \mathbb{R}^{N \times  (q-N)}[s]$ such that $\cap_{i=1}^N \mathcal{B}_i = \big\{ w \in (\mathbb{R}^q)^{\mathbb{Z}_+ }  \mid w = \Pi \mathrm{col}(y,u), \text{ } \eqref{eq:VAR}   \big\}$ for some permutation matrix $\Pi$, and $X$ satisfies Assumption~\ref{ass:SVAR}.
\end{prop}
\begin{proof}
Given minimal kernel representations $R_i$ of $\Sigma_i$, $R= \mathrm{col}(R_1, \dots, R_N)$ has full row rank due to the assumption of regular feedback interconnection, and thus $R$ is a minimal kernel representation of $\Sigma = \land_{i=1}^N \Sigma_i$. Since $n(\Sigma) = \sum_{i=1}^N n(\Sigma_i)$,  following the proofs in \cite[Th. 8 and 12]{willems1997interconnections} analogously, there exists a proper input-output partition $\mathrm{col}(u,y)$ of $w$, where $y=\mathrm{col}(y_1,\dots,y_N)$, leading to a behavioral equation
\begin{equation} \label{eq:proofLemWille}
\begin{bmatrix}
  X_1(\sigma) & \dots & \star\\
  \vdots & \ddots & \vdots \\
  \star & \dots & X_N(\sigma)
\end{bmatrix} \begin{bmatrix}
   y_1\\
   \vdots \\
   y_N
\end{bmatrix}=Q(\sigma) u,
\end{equation}
where $X$ and $Q$ are submatrices of $R$, $X_i$ has dimension $1\times 1$ with $\mathrm{deg}(X_i) = n(\Sigma_i) $, and $\mathrm{deg}(\det(X))  = \sum_{i=1}^N \mathrm{deg}(X_i) = n(\Sigma)$. The last fact shows that the leading row coefficient matrix of $X$ has full rank. Moreover, the fact $\mathrm{deg}(X_i) = n(\Sigma_i) $ shows that the diagonal entries of the leading row coefficient matrix of $X$ equal the leading coefficients of $X_i$. Then pre-multiplying \eqref{eq:proofLemWille} with a real diagonal matrix can transform $X_i$ into a monic polynomial for all $i \in \mathbb{I}_N$, leading to ones on the main diagonal of the leading row coefficient matrix of $X$. This proves the properties in Assumption~\ref{ass:SVAR}.
\end{proof}

If the interconnected system $\land_{i=1}^N \Sigma_i$ in Proposition~\ref{prop:VAR} has an incidence matrix $S$, it is easy to obtain the adjacency matrix of the SVAR model given the obtained $X$ and $Q$ as in the above proof: It holds that $\mathcal{M}([X  \text{ } -Q ]) = S \Pi $. It is also attractive to investigate whether it is possible to obtain the graph of the SVAR model by using only binary matrices, without resorting to polynomial matrices. This will be investigated in future work.

\section{CONCLUSIONS}
 In this work, a behavioral perspective on linear networks is further developed, following the behavioral theory \cite{willems1991paradigms,willems2007behavioral}. This is an initial step towards addressing the two issues, i.e., the incorporation of freedom in input and output selection and building a unified view on different network models.

 Building on the behavioral theory, the novelty of this work lies in the introduction of new graphical representations of behavioral networks by exploiting the concept of hypergraphs. Regularity of interconnections is also discussed and has been shown as an underlying assumption of SVAR models. Moreoever, the connection between the hypergraphs of behavioral networks and the directed graph of SVAR models is investigated.

\section{Acknowledgements}
The first author thanks Timothy H. Hughes, Gareth Willetts from the University of Exeter, and Lizan Kivits from the Eindhoven University of Technology for discussions.

\addtolength{\textheight}{-12cm}  



\bibliographystyle{IEEEtran} 
\bibliography{NetworkSIbehav.bib}

\end{document}